\documentclass[11pt,onecolumn,draftcls]{IEEEtran}

\usepackage{enumerate}
\usepackage{cite}
\usepackage{amsmath}
\usepackage{amsthm}
\usepackage{graphicx}
\usepackage{amssymb}
\usepackage{subfig}
\usepackage{bm,bbm}
\newtheorem{lem}{Lemma}
\newtheorem{thm}{Theorem}
\newtheorem{deft}{Definition}

\newtheorem{rem}{Remark}

\AtBeginDocument{}%
\usepackage{hyperref}

\newcommand{\nozero}{\backslash \left\{\mathbf{0} \right\}}
\newcommand{\I}{\mathbb{I}}

\newcommand{\E}{\mathbb{E}}
\newcommand{\XX}{\mathbf{X}}
\newcommand{\WW}{\mathbf{W}}

\newcommand{\xx}{\mathbf{x}}
\newcommand{\ww}{\mathbf{w}}

\newcommand{\yy}{\mathbf{y}}

\newcommand{\snr}{\text{SNR}}
\newcommand{\eff}{\text{\upshape{eff}}}
\newcommand{\norm}[1]{\left\| #1 \right\|}

\newcommand{\lat}{\Lambda}
\newcommand{\Vor}{\mathcal{V}}
\renewcommand{\vec}[1]{\mathbf{#1}}
\newcommand{\R}{\mathbb{R}}
\newcommand{\pr}[2]{\langle #1, #2 \rangle}
\newcommand{\set}[1]{\left\{#1\right\}}
\newcommand{\eps}{\varepsilon}
\newcommand{\inver}{{\rm err}^{-1}}

\newif\ifnotes\notesfalse

\ifnotes
\usepackage{color}
\definecolor{mygrey}{gray}{0.50}
\newcommand{\notename}[2]{{\textcolor{red}{\footnotesize{\bf (#1:} {#2}{\bf ) }}}}

\newcommand{\anote}[1]{{\notename{Antonio}{#1}}}
\newcommand{\cnote}[1]{{\notename{Cong}{#1}}}
\newcommand{\danote}[1]{{\notename{Daniel}{#1}}}
\newcommand{\dinote}[1]{{\notename{Divesh}{#1}}}

\else

\newcommand{\notename}[2]{{}}

\newcommand{\anote}[1]{}
\newcommand{\cnote}[1]{}
\newcommand{\danote}[1]{}
\newcommand{\dinote}[1]{}

\fi

\begin{document}

\title{\textsc AWGN-Goodness is Enough: Capacity-Achieving Lattice Codes based on Dithered Probabilistic Shaping}

{\author{Antonio Campello, \textit{Member, IEEE}, Daniel Dadush, and Cong Ling, \textit{Member, IEEE}\thanks{
The authors are listed in alphabetical order.

A. Campello and C. Ling are with the Department of Electrical and Electronic Engineering, Imperial College London, UK (e-mail: a.campello@imperial.ac.uk, cling@ieee.org).

D. Dadush is  with  the  Centrum Wiskunde \& Informatika, Netherlands (e-mail: dadush@cwi.nl).

}}

\maketitle
\begin{abstract}
In this paper we show that any sequence of infinite lattice constellations which is good for the unconstrained Gaussian channel can be shaped into a capacity-achieving sequence of codes for the power-constrained Gaussian channel under  lattice decoding and non-uniform signalling. Unlike previous results in the literature, our scheme holds with no extra condition on the lattices (e.g. quantization-goodness or vanishing flatness factor), thus establishing a direct implication between AWGN-goodness, in the sense of Poltyrev, and capacity-achieving codes. Our analysis uses properties of the discrete Gaussian distribution in order to obtain precise bounds on the probability of error and achievable rates. In particular, we obtain a simple characterization of the finite-blocklength behavior of the scheme, showing that it approaches the optimal dispersion coefficient for \textit{high} signal-to-noise ratio. We further show that for \textit{low} signal-to-noise ratio the discrete Gaussian over centered lattice constellations cannot achieve capacity, and thus a shift (or ``dither'') is essentially necessary.
\end{abstract}

\section{Introduction}
Coded modulation schemes for the Gaussian channel can be usually constructed from infinite constellations (ICs) in $\mathbb{R}^n$, shaped according to the power constraint of the channel. In order to analyse such infinite constellations independently of the power, Poltyrev \cite{Poltyrev94} defined the notion of codes which are good for the unconstrained Gaussian channel. In this setting, a code is an infinite discrete subset of $\mathbb{R}^n$, and any point can be transmitted. Since the usual code rate is infinite in this case, the optimal \textit{normalized log density} (NLD) replaces the notion of capacity. The NLD measures the logarithm, per dimension, of the number of points of an IC per unit of volume. An optimal sequence of ICs with vanishing probability of error is called \textit{AWGN-good}, and corresponds to the ``most economic'' constellations that can achieve reliable communication. The most popular ICs are \textit{lattices}, since their symmetries allow for construction of efficient encoding/decoding schemes. Since the work of Poltyrev in 1994, the notion of AWGN-goodness has become an important widely used benchmark and building block for the construction of lattice codes.

Intuitively, AWGN-good ICs should be able to produce capacity-achieving codes
in the power-constrained setting, using nearest neighbor decoding and a
carefully chosen shaping technique. Nevertheless, all known schemes in the
literature that convert lattices into codes for the constrained Gaussian channel
under \textit{lattice decoding} entail some additional property. For instance, Erez and
Zamir \cite{EZ04} proved that an AWGN-good sequence of lattices can be converted
into capacity-achieving codes for the Gaussian channel, \textit{provided that it
can be nested to another lattice sequence} which is also AWGN-good and has
optimal normalized second-moment (i.e., quantization-good). 
Recently, Ling and Belfiore \cite{LB14} have
shown that a sequence of lattices along with probabilistic shaping can achieve
the capacity of the Gaussian channel above a certain signal-to-noise ratio (SNR) \anote{$\snr = e$ is not really high...}
\textit{provided that it has vanishing flatness factor}, which relates to how quickly a
Gaussian random vector becomes equidistributed over cosets of the lattice as its
standard deviation increases. 

The objective of this paper is to show that indeed AWGN-goodness is a
\textit{sufficient} property for building capacity-achieving lattices in the Gaussian channel, with no extra condition. In order to do so, we employ a technique based on non-equiprobable signalling using discrete Gaussian distributions centered at a (non-zero) randomly generated point in $\mathbb{R}^n$. Following \cite{PZ12}, we call such technique \textit{dithered probabilistic shaping} (DPS). From a practical ``separation'' point of view, this means that the design problem of good lattices for the AWGN channel can be completely \textit{decoupled}: one can focus entirely on the design problem for the unconstrained channel, which can be then coupled with \textit{plug-and-play} DPS techniques into a good code for the constrained channel.


\subsection{Main Result}
\begin{figure}[!ht]
	\centering
	\subfloat[A one-dimensional discrete Gaussian with support in $\mathbb{Z}+0.3$]{\includegraphics[scale=0.5]{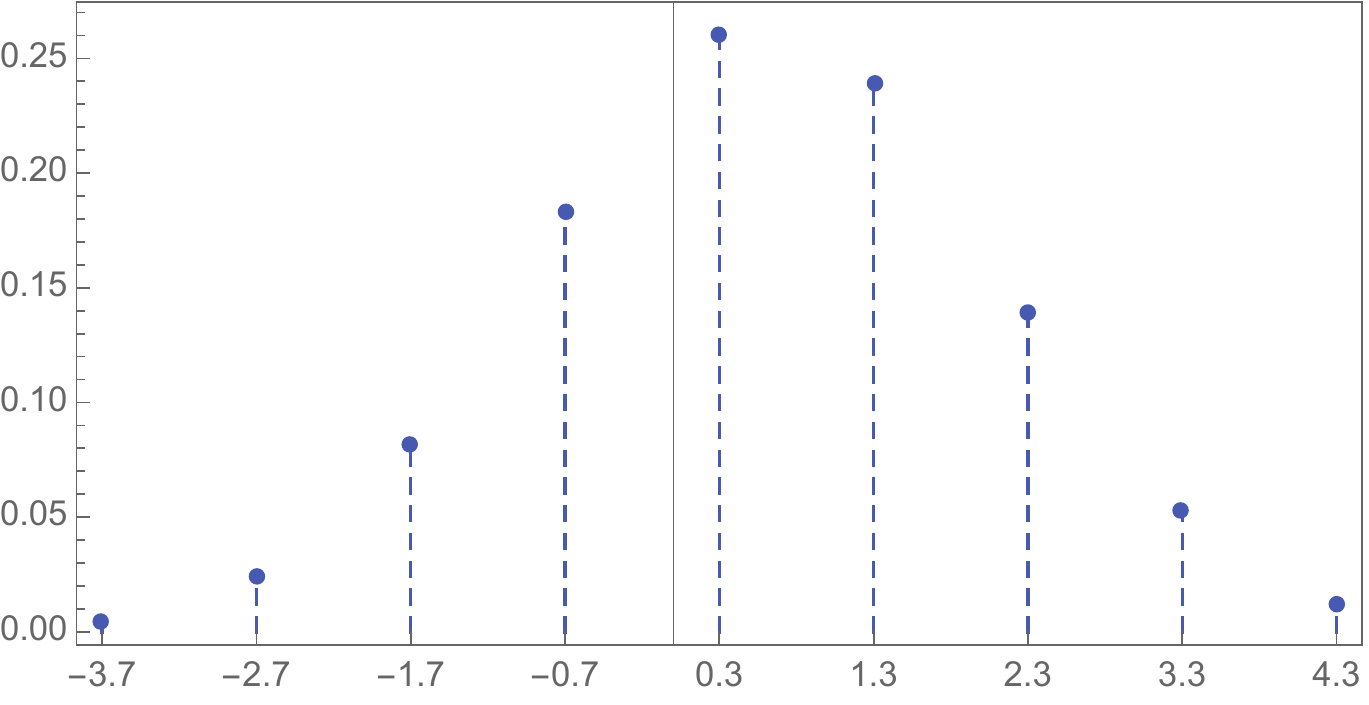}} \hfill
	\subfloat[A two-dimensional centered discrete Gaussian with support in a hexagonal lattice]{\includegraphics[scale=0.6]{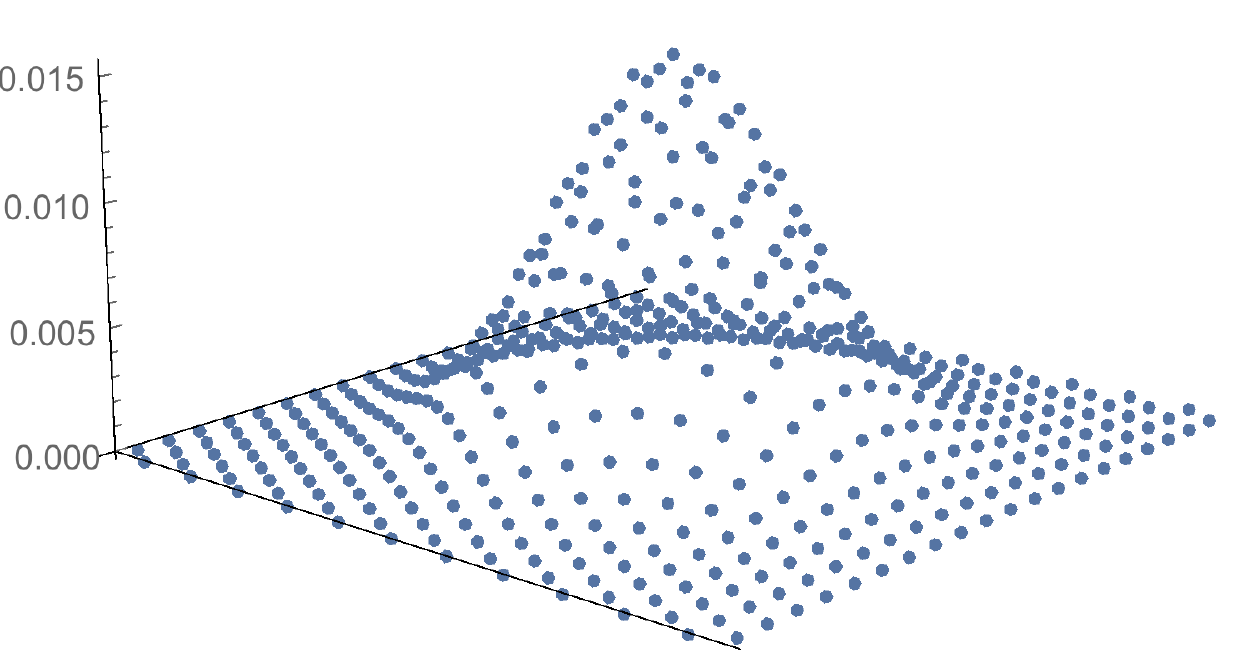}}
	\caption{Illustration of the discrete Gaussian distribution}
	\label{fig:unique}
\end{figure}
 Let $\vec{W} \sim \mathcal{N}(\vec{0}, \sigma_w^2)$ be an $n$-dimensional Gaussian random vector. A sequence of \textit{AWGN-good} lattices $(\lat_n)_{n=1}^{\infty}$ with volume $V(\lat_n)$ is such that the probability of error (i.e. the probability that $\vec{W}$ leaves the Voronoi cell of $\lat_n$) vanishes for any NLD
\begin{equation} \delta(\lat_n) \triangleq -\frac{1}{n} \log V(\lat_n) < -\frac{1}{2} \log (2 \pi e \sigma_w^2) \triangleq \delta^*.
\label{eq:poltyrevLimit}
\end{equation}
The discrete Gaussian distribution $\mathcal{D}_{\lat+\mathbf{t},\sigma_s^2}$ is the distribution taking values in $\lat+\mathbf{t}$ whose mass of $\xx + \mathbf{t}$ is proportional to $e^{-\left\| \xx + \mathbf{t}\right\| /2\sigma_s^2}$ (see Figure \ref{fig:unique}). This distribution has finite power and can thus be used in the transmission over a Gaussian channel with average power constraint $P$. In a lattice Gaussian coding scheme, the sent signal is chosen according to $\mathcal{D}_{\lat+\mathbf{t},\sigma_s}$, while the received points are suitably scaled and then decoded using lattice decoding.

Let $(\lat_n)_{n=1}^\infty$ be a sequence of AWGN-good lattices. Our main result states the existence of a sequence of shift vectors (or ``dithers") $(\mathbf{t}_n)_{n=1}^\infty$ such that the lattice coding scheme with distribution $\mathcal{D}_{\lat_n+\mathbf{t}_n,\sqrt{P}}$ is capacity-achieving in the Gaussian channel.
A more quantitative statement can be found in Theorem
\ref{thm:awgn-good-enough}, along with an analysis of the rate of convergence to
capacity. For instance, we show that the existence of $\mathbf{t}_n$ holds with constant probability and, using \cite{IZF13}, that the gap to
capacity has order $\approx (2n)^{-1/2}(Q^{-1}(\varepsilon) + \sqrt{8})$ where
$Q^{-1}$ is the inverse-error function. Apart for the $\sqrt{8}$ term (which we
believe is an artifact of our analysis), this corresponds to the optimal
convergence behavior for the Gaussian channel for high SNR \cite{PPV09}.
In particular, as $\varepsilon \to 0$, the optimal dispersion is approached at high SNR.

A key point for the proof of the main results is the \textit{sampling lemma} (Lemma \ref{lem:sampling}), that says that the dithered lattice Gaussian is equivalent to a continuous Gaussian. Since continuous dithers would require sharing an infinite number of random bits, we then demonstrate the sufficiency of discrete dithering in Lemma \ref{lem:discrete-sampling-lemma}.
As a further point of interest, we observe using~\cite{CDLP13} that the coding
gain of a lattice is closely related to the so-called \emph{smoothing parameter} of the dual lattice.  At a high level, the discrete Gaussian over dual lattice
$\lat^*$ sampled above the smoothing parameter ``looks like'' a
continuous Gaussian, which justifies the name. The study of this parameter has
played a fundamental role in our understanding of the discrete Gaussian, and has
led to many important developments such as tight transference theorems in the
geometry of numbers~\cite{Banaszczyk93}, new lattice based cryptographic
schemes~\cite{MR04,Regev05,GPV08}, and the recent development of the so-called
reverse Minkowski inequality~\cite{DR16,RS17}. This relation, discussed in
Section~\ref{sec:discussion}, gives an alternative viewpoint for quantifying
AWGN-goodness that we
hope will find future use.


\subsection{To Dither or Not to Dither}
Since the seminal work \cite{EZ04}, nested lattice constellations are known to
be capacity achieving in the Gaussian channel with lattice decoding,
provided that the constellations are suitably shifted by a vector known both at the transmitter and at the receiver. The process of randomizing the choice of the shift vector, known as \textit{dithering}, greatly facilitates the analysis but may introduce additional design complexity. An intriguing open
problem is whether dithering is indeed necessary. For Voronoi constellations (under MMSE scaled lattice
decoding), this problem is considered in \cite{YH17,ZamirBook}, where the
necessity of a dither is argued\anote{argued?  shown? discussed?}\danote{shown I
would say, the note was pretty convincing, but I guess qualify that it's for the
``standard'' scheme somehow} for low SNR.

Recently, Ling and Belfiore \cite{LB14} have shown that undithered lattice
coding, with probabilistic shaping according to a discrete Gaussian
distribution, is capacity achieving for a threshold\footnote{In fact, the arguments in \cite{LB14} can be slightly improved in
order to reduce the threshold to $\text{SNR} > e-1$} SNR greater than $e$. In another work \cite{PietroZB18}, di Pietro, Zemor and Boutros showed that LDA lattices without dithering achieve capacity on the {Gaussian} channel if SNR $>1$. In
Section \ref{sec:converse} of the present paper, we show that (a non-zero) dither is indeed necessary in the low SNR regime. Specifically we show (Theorem \ref{thm:converse}) that if a sequence of lattices is shaped according to a \textit{centered} Gaussian distribution with variance parameter equals to the power constraint of the channel, it fails to achieve any positive rates for $\text{SNR} < 1$. \anote{I moved a sentence from the first section to here, to avoid redundancy} Interestingly, in the low SNR regime our scheme uses properties of the
discrete Gaussian when the flatness factor is large (or ``below smoothing'', in the computer science jargon),
a setting where previous analytical techniques break down. As it turns out, in this regime, dithering is not only a matter of simplification but essentially necessary. It is noting that although polar lattices \cite{YanLLW14} can achieve capacity for all SNRs, their construction employ randomly generated bits which, ultimately, play the role of a dither.

\subsection{Related Works}
In comparison to previous works on nested lattice constellations \cite{EZ04}, approaches based on discrete Gaussians do not rely on a quantization-good sub-lattice for shaping. It is worth noting that, since  \cite{EZ04}, there has been a significant effort employed in simplifying nested-lattice constellations. For instance, \cite{OE16} removes the assumption that the shaping sub-lattice needs to be AWGN-good and \cite{QFH17} by-passes the quantization-good by using a discrete dither and removing the constellation points with large power. Relying on the machinery of the discrete Gaussian distribution, our approach removes the reliance on a sub-lattice with suitable properties, and greatly simplifies the proof that lattice codes are capacity-achieving in the Gaussian channel (our full proof is self-contained in Section \ref{sec:AWGN-good}).

Our main result strongly uses non-uniform signalling, which could be implemented with probabilistic shaping techniques. In a broader scope, besides the theoretically appealing properties of probabilistic shaping for Gaussian channels \cite{Frank93},\cite{LB14}, the techniques became  popular in the last years due to their prospective applications to non-linear optical-fiber transmissions \cite{Fehenberger16},  \cite{SK12}. The idea of \textit{dithered probabilistic shaping} (DPS) was previously considered in \cite{PZ12}, who proposed a method to achieve high shaping gains of a constellation with small dimension and low complexity.

Besides showing the sufficiency of the AWGN goodness figure of merit, our analysis closes the gap of uniform signalling \cite{LB14} for low SNR. Although this is not the usual regime for wireless channels, it has attracted a recent interest due to secrecy applications in low profile covert communications \cite{WWZ16,Bloch16,Bash13}, where signals are transmitted with very low power. Understanding how probabilistic shaping schemes behave for low SNR might find use in such scenarios. \anote{Removed the former paragraph, and added something about applications to IT}In addition to AWGN channel coding, there is an increasing literature in the use of the discrete Gaussian in other information-theoretic scenarios where our simplified approach could play a role. A few examples are Gaussian wiretap channels \cite{LLBS_12}, fading wiretap channels \cite{LLV16}, and compound channels \cite{CLB16}.


%

\section{Notation and Preliminary Results}
We consider a real-valued AWGN channel with average-power constraint $P$ and noise variance $\sigma_w^2$. Denote the received signal of the channel, given input $\xx \in \mathbb{R}^n$, by $\yy = \xx + \ww$, where $\ww$ is drawn from the distribution $\mathcal{N}(0,\sigma_w^2 \mathbf{I}_n)$.
A (full-rank) lattice $\Lambda \subset \mathbb{R}^n$ is a discrete subgroup not contained in any proper subspace of $\mathbb{R}^n$.

\begin{deft} A lattice code for the Gaussian channel consists of
	\begin{enumerate}
		\item A lattice $\Lambda \subset \mathbb{R}^n$.
			\item A finite-entropy probability distribution $\mathcal{D}$ taking values in $\Lambda$ such that, for $\XX \sim \mathcal{D}, E[\left\| \XX \right\|^2] \leq nP.$
			\item A decoding function $g:\mathbb{R}^n \to \Lambda$.
	\end{enumerate}
	\label{def:latticeCode}
\end{deft}
The error probability of a lattice code is $P_e(\Lambda) \triangleq \Pr(g(\mathbf{X}+\mathbf{W})\neq \mathbf{X})$, where $\mathbf{W} \sim \mathcal{N}(0,\sigma_w^2 \mathbf{I}_n)$. The \textit{rate} of the code is given by $R = (1/n) \mathbb{H}(\XX),$
where $$\mathbb{H}(\XX) \triangleq -\sum_{\xx \in \Lambda} \Pr(\XX=\xx) \log \Pr(\XX=\xx)$$ is the entropy of distribution $\mathcal{D}$ (all logarithms are with respect to base $e$ and rates are calculated in \textit{nats}). We can similarly define lattice codes for a translation $\Lambda + \mathbf{t}$, where $\mathbf{t} \notin \Lambda$, with the obvious modifications.

If $\mathcal{D}$ is the uniform distribution over the points of $\Lambda$ in a compact set $\mathcal{S}$ (e.g. a ball or the Voronoi region of a sub-lattice), and zero otherwise, this corresponds to the classic ``deterministic'' shaping. In this case, a set of possible uniformly distributed messages $\left\{1, 2, \ldots, e^{nR} \right\}$ can be mapped into the signal space by simply labeling the points in $\Lambda \cap \mathcal{S}$. When $\mathcal{D}$ is non-uniform, messages can be mapped into the signal space by means of \textit{probabilistic shaping} techniques (e.g. \cite{PZ12}, \cite[Sec IV]{LB14}\cite{Bocherer14}, \cite[Sec 6.5]{ZamirBook}).

Let $\text{SNR} \triangleq P/\sigma_w^2$ and $C(\text{SNR}) \triangleq 1/2 \log(1+\text{SNR})$.
\begin{deft}
	A sequence of lattice codes of increasing dimension $\Lambda_n$ with distributions $\mathcal{D}_n$ is said to be capacity-achieving in the Gaussian channel with $\text{\upshape{SNR}} = P/\sigma_w^2$ if $P_e(\Lambda_n) \to 0$ and
	\begin{equation}
	\lim_{n\to\infty} \frac{1}{n}\mathbb{H}(\XX_n) = C(\text{\upshape{SNR}}),
	\end{equation}
	where $\XX_n \sim \mathcal{D}_n$.
\end{deft}

  \begin{deft}
 	Let
 	$$f_{\sigma}(\xx) = \frac{e^{-\left\|\mathbf{x}\right\|^2/2\sigma^2}}{(\sqrt{2\pi\sigma^2})^n}.$$
 	 The \textit{discrete Gaussian distribution} $\mathcal{D}_{\Lambda+\mathbf{t},\sigma}$ is the discrete distribution assuming values in $\Lambda+\mathbf{t}$ such that the mass of each point $\xx+\mathbf{t} \in \Lambda+\mathbf{t}$ is proportional to $f_{\sigma }(\xx+\mathbf{t})$.
 \end{deft}

 If $\mathcal{D}$ in Definition \ref{def:latticeCode} is a discrete Gaussian distribution assuming values in $\Lambda + \vec{t}$, the corresponding code is called a \textit{lattice Gaussian code}. Some discrete Gaussians are illustrated in Figure \ref{fig:unique}. If a point in a lattice Gaussian code is transmitted over a Gaussian channel, we have the following fact:
\begin{lem}[Equivalence between MAP and lattice decoder \cite{LB14}] Let
 $\hat{\xx}$ be the output of the maximum-a-posteriori decoder for a lattice Gaussian code with distribution $\mathcal{D}_{\Lambda+\mathbf{t},\sigma_s^2}$. We have
	\begin{equation*}
	\hat{\xx} = \displaystyle \arg\min_{\xx \in \Lambda + \mathbf{t} }  \left\| \alpha \yy - \xx \right\|,
	\end{equation*}
	where $\alpha = \sigma_s^2 /(\sigma_s^2+\sigma_w^2)$ is called the \text{\upshape{MMSE}} (or Wiener) coefficient.
	\label{lem:MAP}
\end{lem}
Therefore, the optimal MAP decoder is obtained by MMSE pre-processing followed by lattice decoding. Let $\vec{W}_{\eff} = (\alpha-1) \vec{X} + \alpha \vec{W}$ be the effective noise in this process.
For an $n$-dimensional lattice $\lat \subset \R^n$, we define the Voronoi cell
$\Vor(\lat)$ of $\lat$ by
\[
\Vor(\lat) = \set{\vec{x} \in \R^n: \pr{\vec{x}}{\vec{y}} \leq \|\vec{y}\|^2/2 \text{, for any } \yy \in \Lambda}
\]
namely, the set of all points closer to the origin than any other lattice point. The volume of $\Lambda$, denoted by $V(\Lambda)$, corresponds to the volume of the Voronoi cell $\mathcal{V}(\Lambda)$. From Lemma \ref{lem:MAP}, the probability of error of the maximum-a-posteriori decoder is
$$P_{e}(\Lambda) = \Pr(\vec{W}_{\eff} \notin \mathcal{V}(\Lambda)).$$
We denote by $\mathbf{x}  \mbox{ mod} \Lambda$ the unique representative of $\mathbf{x} \in \mathbb{R}^n$ in the voronoi cell of $\Lambda$, with ties broken arbitrarily.

\section{Capacity Achieving Codes From AWGN-Good Lattices}
\label{sec:AWGN-good}
In this section, we will show that any AWGN-good lattice $\lat$ can be used to
construct a capacity achieving code for the AWGN-channel. The caveat will be the
use of a random dither as well as the use of non-equiprobable encoding, namely discrete
Gaussian encoding. We will later argue in in Section \ref{sec:discussion} that a discrete dither suffices and in Section \ref{sec:converse} that the dither is essentially necessary for this construction to hold for all SNR.

%
%
%


Let us now examine an ensemble of coding schemes in $\R^n$ based on a
$n$-dimensional lattice $\lat \subset \R^n$. Define the inverse error function
of $\lat$ by
\begin{equation}
\inver_\eps(\lat) = \min \set{s: \Pr_{\vec{x} \sim
\mathcal{N}(\vec{0},\mathbf{I}_n)}[x \notin s\Vor(\lat)] \leq \eps} \text{ .}
\label{eq:errInv}	
\end{equation}
Note that for a AWGN-good family $(\Lambda_n)_{n=1}^\infty$, $V(\Lambda_n) =
1 ~ \forall i \in [n]$, for any fixed $\eps > 0$ we have that $\lim_{n
\rightarrow \infty} \inver_\eps(\lat) \rightarrow \sqrt{2\pi e}$. We also define the
\textit{normalized volume-to-noise ratio} (NVNR) of lattice at probability of error $\eps$ as:
\begin{equation} \mu(\Lambda,\eps) =
{\inver_\eps(\lat)^2 V(\lat)^{2/n}}.
\label{eq:NVNR}
\end{equation}
This definition indeed corresponds to the usual NVNR (e.g. \cite[Def. 3.3.3]{ZamirBook}), rephrased in light of the error function \eqref{eq:errInv}. For convenience, we further normalize $\mu(\Lambda,\eps)$ by $2\pi e$ and define
\begin{equation} \gamma(\Lambda,\eps) = \frac{\mu(\Lambda,\eps)}{2\pi e}.
\label{eq:NNVNR}
\end{equation}
This ratio can be interpreted as the ``modulation loss'' of $\Lambda$. For an AWGN-good sequence, $\log \gamma(\Lambda_n,\eps)$ vanishes as $n \to \infty$.

Let $\sigma_s^2$, $\sigma_w^2$ denote the nominal power of the signal and the noise power per symbol
respectively and let $\sigma_{{\rm eff}}^2 := \frac{\sigma_s^2
\sigma_w^2}{\sigma_s^2 + \sigma_w^2}$.

Let $\lat \subset \R^n$ be a lattice. Given a fixed probability of error $\eps > 0$, we
will now examine a family of dithered coding schemes which on average have
decoding error at most $6\eps$, power per symbol $\leq \sigma_s^2(1 + 4/\sqrt{n})$,
and gap to the capacity depending only on the relationship between the
inverse error function and the lattice determinant. The existence of a good code
from our family will then follow by the union bound appropriately applied.

For the purpose of finding a good code, we shall examine the family of encoding
distributions $D_{\lat+\vec{t}, \sigma_s}$, indexed by $\vec{t} \in \R^n$, i.e.~the
discrete Gaussian on $\lat+\vec{t}$ of parameter $\sigma_s$. For a given
$\vec{t}$, the decoding function $g_\vec{t}$ will be maximum a posteriori decoding
as in Lemma~\ref{lem:MAP}, namely for a noisy signal $\vec{y} \in \R^n$, the
decoded signal corresponds to
\[
g_\vec{t}(\vec{y}) := \arg \min_{\vec{x} \in \lat+\vec{t}} \left\|\vec{x}-\frac{\sigma_s^2}{\sigma_s^2+\sigma_w^2}
\vec{y}\right\| \text{ .}
\]

To prove our bounds on the average properties of such codes we need a
distribution on shifts $\vec{t}$. For this distribution, we pick the natural
choice $\vec{T} \sim \mathcal{N}(\vec{0}, \sigma_s^2 \mathbb{I}_n)$, i.e.~the
shift is distributed according to the maximum entropy input distribution
satisfying the average power constraint.

The main theorem we will prove in this section is as follows.
\begin{thm}
Let $\lat \subset \R^n$ be an $n$-dimensional lattice, $n \geq 16$. For
power $\sigma_s^2=P$, noise variance $\sigma_w^2 > 0$ and probability of error $\eps > 0$,
let $\sigma_{\rm eff}^2 = \frac{\sigma_s^2 \sigma_w^2}{\sigma_s^2+\sigma_w^2}$,
$s_\eps= \inver_\eps(\lat) \cdot \sigma_{\rm eff}$ and $\gamma(\Lambda,\eps) =
\inver_\eps(\lat)^2 V(\lat)^{2/n}/(2\pi e)$.
Then for $\vec{T} \sim \mathcal{N}(\vec{0},\sigma_s^2 \mathbb{I}_n)$ with probability at least $1/2$ the
coding distribution $\vec{X} \sim D_{s_\eps \lat + \vec{T}, \sigma_s}$ equipped with maximum
likelihood decoding satisfies:
\begin{enumerate}
\item The decoding error is bounded by $6\eps$.
\item The squared power per symbol is at most $(1+4/\sqrt{n})\sigma_s^2 $.
\item The per symbol gap to capacity is bounded by $\frac{1}{2}\log \gamma(\Lambda,\eps) +
\frac{2}{\sqrt{n}} + \frac{4}{n}$.
\end{enumerate}
\label{thm:awgn-good-enough}
\end{thm}
\anote{Moved the notational convenience for after the statement of the theorem}
For notational convenience, we will assume in the remainder that $\sigma_{\rm
	eff} \cdot \inver_\eps(\lat) = 1$ (note that this can be achieved by simply
scaling the lattice). With this normalization, we will be able to achieve a low
probability of error by choosing the codes directly from $\lat$ (instead of a scaling of $\lat$).

To prove the theorem, we will rely on the subsequent lemmas which characterize
the behavior of a randomly dithered channel. We first present these lemmas and
prove Theorem~\ref{thm:awgn-good-enough} at the end of the section.

The following simple sampling lemma shows the random variables corresponding to
either sampling $\vec{T}$ from $\mathcal{N}(\vec{0}, \sigma_s^2 \mathbb{I}_n)$
and a) returning $\vec{T}$ or b) returning a discrete Gaussian sample from
$D_{\lat+\vec{T},\sigma_s}$ have the same distribution.

\begin{lem}[Sampling Lemma]
Let $\vec{T} \sim \mathcal{N}(\vec{0},\sigma_s^2 \mathbb{I}_n)$ and let
$\vec{X} \sim D_{\lat+\vec{T},\sigma_s}$. Then $\vec{T}$ and $\vec{X}$ are identically distributed.
\label{lem:sampling}
\end{lem}
\begin{proof}
We need only show that $\vec{X}$ has the same probability density as $\vec{T}$. For $\vec{w}
\in \R^n$, note that $\vec{X}$ can only hit $\vec{w}$ if $\vec{w} \in
\lat+\vec{T} \Leftrightarrow \vec{T} \in \lat+\vec{w}$. Therefore, for any measurable set $A
\subseteq \R^n$, we have that
\begin{align*}
\Pr[\vec{X} \in A] &= \int_{\R^n} f_{\sigma_s}(\vec{t}) \Pr[\vec{X} \in A | T = \vec{t}] { d\mathbf{t}} \\
      &= \int_{\R^n} f_{\sigma_s}(\vec{t})
       \frac{f_{\sigma_s}(A \cap (\lat+\vec{t}))}{f_{\sigma_s}(\lat+\vec{t})} { d\mathbf{t}} \\
      &= \int_{\mathcal{V}(\lat)} \sum_{\vec{x} \in \lat+\vec{c}} f_{\sigma_s}(\vec{x})
       \frac{f_{\sigma_s}(A \cap (\lat+\vec{c}))}{f_{\sigma_s}(\lat+\vec{c})} { d\mathbf{c}} \\
      &= \int_{\mathcal{V}(\lat)} f_{\sigma_s}(A \cap (\lat+\vec{c})) { d\mathbf{c}}
      = \int_A f_{\sigma_s}(\vec{x}) { d\mathbf{x}} \text{, }
\end{align*}
as needed.
\end{proof}
\begin{rem}
	Note that the sampling lemma still holds if we take $\vec{T} \sim \mathcal{N}(0,\sigma_s^2\I_n) \text{ \upshape{mod} } \Lambda$, where the mod-$\Lambda$ operation maps a point in $\mathbb{R}^n$ to a coset representative in the Voronoi cell of $\Lambda$. The $\ell_{\infty}$ distance between $\mathcal{N}(0,\sigma_s^2\I_n) \text{ \upshape{mod} } \Lambda$ and a uniform distribution in $\mathcal{V}(\Lambda)$, normalized by $1/V(\Lambda)$, is the so-called flatness factor (e.g. \cite{LB14}). If the flatness factor of $\Lambda$ is small, roughly speaking, one could replace $\vec{T}$ in Lemma \ref{lem:sampling} by a uniform dither. As we will see later in Section \ref{sec:converse}, this condition is too stringent for low SNR.
\end{rem}

We now show that expected probability of error of our lattice family is exactly $\eps$.  In
what follows, $\vec{T} \sim \mathcal{N}(\vec{0},\sigma_s^2 \mathbb{I}_n)$ will
be our random shift of $\lat$, $\vec{W} \sim \mathcal{N}(\vec{0},\sigma_w^2
\mathbb{I}_n)$ will be the channel noise, $\vec{X} \sim
D_{\lat+\vec{T},\sigma_s}$ will denote the coding distribution, and
$g_\vec{T}$ is our decoding function. We recall our assumption that
$\sigma_{{\rm eff}} \cdot \inver_\eps(\lat) = 1$.

\begin{lem}[Error Probability Bound] For $\gamma \geq 1$,
\[
\Pr_\vec{T}[\Pr_\vec{X}[g_{\vec{T}}(\vec{X}+\vec{W}) \neq \vec{X}] \geq \gamma \eps] \leq 1/\gamma \text{ .}
\]
\label{lem:decoding-error}
\end{lem}
\begin{proof}
Let $\alpha = \frac{\sigma_s^2}{\sigma_s^2+\sigma_w^2}$.
Recall that $g_{\vec{T}}(\vec{X}+\vec{W}) \neq \vec{X}$ if and only if
\[
(1-\alpha)\vec{X} + \alpha \vec{W} \notin \mathcal{V}(\lat) \text{ .}
\]
By the sampling lemma $\vec{X}$ has distribution
$\mathcal{N}(\vec{0},\sigma_s^2 \mathbb{I}_n)$. Given that $\vec{X}$ and $\vec{W}$ are
independent Gaussians, we have that $(1-\alpha)\vec{X}+ \alpha \vec{W}$ is distributed as
$\mathcal{N}(\vec{0},\sigma_{\rm eff}^2 \mathbb{I}_n)$ since by construction
\[
\sigma_{{\rm eff}}^2 := \frac{\sigma_s^2 \sigma_w^2}{\sigma_s^2+\sigma_w^2}
                      = (1-\alpha)^2 \sigma_s^2 + \alpha^2 \sigma_w^2 \text{ .}
\]
By assumption on $\lat$ we have that $\sigma_{\rm eff} \cdot \inver_\eps(\lat) = 1$, and thus
\begin{align*}
\Pr[(1-\alpha)\vec{X} + \alpha \vec{W} \notin \mathcal{V}(\lat)]
&= \Pr_{\vec{Z} \sim \mathcal{N}(\vec{0},\mathbb{I}_n)}[\sigma_{\rm eff} \vec{Z} \notin \mathcal{V}(\lat)] \\
&= \Pr_{\vec{Z} \sim \mathcal{N}(\vec{0},\mathbb{I}_n)}[\vec{Z} \notin \inver_{\eps}(\lat)
\mathcal{V}(\lat)] = \eps \text{ .}
\end{align*}

Using the above, by Markov's inequality we have that
\[
\Pr_{\vec{T}}[\Pr_\vec{X}[g_{\vec{T}}(\vec{X}+\vec{W}) \neq \vec{X}] \geq \gamma \eps] \leq
\frac{\Pr[g_{\vec{T}}(\vec{X}+\vec{W}) \neq \vec{X}]}{\gamma \eps} = \frac{1}{\gamma} \text{ ,}
\]
as needed.
\end{proof}

The next lemma shows that average power per symbol is very close to the desired
limit. The proof proceeds by a comparison to the continuous Gaussian followed by
a standard Chernoff bound.

\begin{lem} For any $\eps > 0$, we have that
\begin{align*}
1. \Pr_\vec{T}[\E[\|(\vec{X}/\sigma_s)\|^2 | \vec{T}]] \geq (1+\eps)n] \leq e^{-(\eps^2/4-\eps^3/6)n} \\
2. \Pr_\vec{T}[\E[\|(\vec{X}/\sigma_s)\|^2 | \vec{T}]] \leq (1-\eps)n] \leq e^{-(\eps^2/4+\eps^3/6)n}
\end{align*}
\label{lem:concentration}
\end{lem}
\begin{proof}
We first prove the upper bound. Since $\vec{X}/\sigma_s \sim
\mathcal{N}(0,\mathbb{I}_n)$, a standard computation reveals that $\E[e^{\alpha\|\vec{X}/\sigma_s\|^2}] = (1-2\alpha)^{-n/2}$ for $\alpha < 1/2$. By the
Chernoff bound

\begin{equation*}
\begin{split}
\Pr_\vec{T}[\E[\|(\vec{X}/\sigma_s)\|^2| \vec{T}] \geq (1+\eps)n]
&\leq \min_{\alpha \in (0,1/2)} \frac{\E_\vec{T}[e^{\alpha
\E[\|(\vec{X}/\sigma_s)\|^2| \vec{T}]}]}{e^{\alpha(1+\eps)n}} \\
&\leq \min_{\alpha \in (0,1/2)} \frac{\E[e^{\alpha \|(\vec{X}/\sigma_s)\|^2}]}
      {e^{\alpha(1+\eps)n}} \quad \left(\text{ by Jensen's inequality }\right) \\
&= \min_{\alpha \in (0,1/2)} \left(\frac{e^{-(1+\eps)\alpha}}{\sqrt{1-2\alpha}}\right)^{n} \\
&=
(\sqrt{1+\eps}e^{-\eps/2})^n < e^{-(\eps^2/4-\eps^3/6)n}
      \left( \text{ setting $\alpha = \frac{\eps}{2(1+\eps)}$ } \right) \text{ .}
\end{split}
\end{equation*}
Similarly, for the lower bound, we have
\begin{align*}
\Pr_\vec{T}[\E[\|(\vec{X}/\sigma_s)\|^2| \vec{T}] \leq (1-\eps)n]
&\leq \min_{\alpha > 0} \frac{\E_\vec{T}[e^{-\alpha \E[\|(\vec{X}/\sigma_s)\|^2|
\vec{T}]}]}{e^{-\alpha(1-\eps)n}} \\
&\leq \min_{\alpha > 0} \frac{\E[e^{-\alpha \|(\vec{X}/\sigma_s)\|^2}]}
      {e^{\alpha(1-\eps)n}} \quad \left(\text{ by Jensen's inequality }\right) \\
&= \min_{\alpha > 0} \left(\frac{e^{(1-\eps)\alpha}}{\sqrt{1+2\alpha}}\right)^{n} \\
&=
(\sqrt{1-\eps}e^{\eps/2})^n < e^{-(\eps^2/4+\eps^3/6)n}
      \left( \text{ setting $\alpha = \frac{\eps}{2(1-\eps)}$ } \right) \text{ .}
\end{align*}
\end{proof}

We now argue that the entropy of the coding distribution is large with good
probability. In particular, we would like to know that for most choices
$\vec{t}$ for $\vec{T}$ that $\mathbb{H}(\vec{X} | \vec{T}=\vec{t})$ is large.
Recall that $\vec{X} | \vec{T}=\vec{t}$ is distributed as
$D_{\Lambda+\vec{t},\sigma_s}$, where a direct computation gives
\anote{Changed the notation of the entropy to "uniformise" with the rest of the text}
\begin{equation}
\label{eq:entropy}
\begin{split}
\mathbb{H}(\vec{X} | \vec{T} = \vec{t})
&= \sum_{\vec{y} \in \lat+\vec{t}}
-\log\left(\frac{f_{\sigma_s}(\vec{y})}{f_{\sigma_s}(\lat+\vec{t})}\right)
\frac{f_{\sigma_s}(\vec{y})}{f_{\sigma_s}(\lat+\vec{t})} \\
&=
\log\left((\sqrt{2\pi \sigma^2_s})^n f_{\sigma_s}(\lat+\vec{t})\right)
+ \frac{1}{2} \E_{\vec{Z} \sim D_{\lat+\vec{t},\sigma_s}}[\|\vec{Z}/\sigma_s\|^2] \text{ .}
\end{split}
\end{equation}
Thus, to prove that the entropy is large we must show that both
$f_{\sigma_s}(\lat+\vec{t})$ and $\E[\|\vec{Z}/\sigma_s\|^2]$ are large with good
probability over $\vec{t}$. Note that the latter condition is essentially given
by Lemma~\ref{lem:concentration}, so we focus now on the former.

The following lemma uses a bound on the negative moment of
$f_{\sigma_s}(\lat+\vec{t})$ to show that it is unlikely to be too small.

\begin{lem}
\[
\Pr_{\vec{T}}\left[f_{\sigma_s}(\lat+\vec{T}) \leq \frac{\eps}{V(\lat)}\right] \leq \eps \text{ .}
\]\label{lem:gaussian-lb}
\end{lem}
\begin{proof}
To begin, we have that

\begin{align*}
\E_\vec{T}[f_{\sigma_s}(\lat+\vec{T})^{-1}]
&= \int_{\R^n} f_{\sigma_s}(\vec{t}) f_{\sigma_s}(\lat+\vec{t})^{-1} { d\mathbf{t}}
= \int_{\R^n / \lat} \sum_{\vec{t} \in \lat+\vec{c}} f_{\sigma_s}(\vec{t}) f_{\sigma_s}(\lat+\vec{c})^{-1} { d\mathbf{c}} \\
&= \int_{\R^n / \lat} f_{\sigma_s}(\lat+\vec{c}) f_{\sigma_s}(\lat+\vec{c})^{-1} { d\mathbf{c}}
= \int_{\R^n / \lat} { d\mathbf{c}} = V(\lat) \text{ .}
\end{align*}
By Markov's inequality, we have
\begin{align*}
\Pr_{\vec{T}}\left[f_{\sigma_s}(\lat+\vec{T}) \leq \frac{\eps}{V(\lat)}\right] \leq
\eps ~ \frac{\E_{\vec{T}}\left[f_{\sigma_s}(\lat+\vec{T})^{-1}\right]}{V(\lat)} = \eps \text{ ,}
\end{align*}
as needed.
\end{proof}
\anote{I think d) should be $\leq$. Also, if we are willing to lose $\log n/n$ (which we are going to lose anyway), we can replace $e^4$ by $n$, improving the probability of a good dither.}
\danote{What error do you see? I'm not sure that we can do much to improve the
probability of the dither being good, because this is only one of many events we need
to control. In particular, we would need tighter concentration for the error
probability in Lemma~\ref{lem:decoding-error} than we have now.}

\begin{proof}[Proof of Theorem~\ref{thm:awgn-good-enough}]
	By scaling the lattice, we may assume as above that $\sigma_{\rm eff} \cdot
	\inver_\eps(\lat) = 1$ and thus that $s_\eps = 1$. By
	Lemmas~\ref{lem:decoding-error},~\ref{lem:concentration},~\ref{lem:gaussian-lb}
	we have that
	\begin{enumerate}
		\item[a.] $\Pr_\vec{T}[\Pr_\vec{X}[g_{\vec{T}}(\vec{X}+\vec{W}) \neq \vec{X}] \geq 6 \eps] \leq 1/6$.
		\item[b.] $\Pr_\vec{T}[\E[\|\vec{X}/\sigma_s\|^2 | \vec{T}] \geq n + 4\sqrt{n}] \leq e^{-4 +
			64/(6\sqrt{n})} \leq e^{-4/3}$ for $n \geq 16$ (setting $\eps = 4/\sqrt{n}$).
		\item[c.] $\Pr_\vec{T}[\E[\|\vec{X}/\sigma_s\|^2 | \vec{T}] \leq n - 4\sqrt{n}] \leq e^{-4}$ (setting $\eps = 4/\sqrt{n}$).
		\item[d.] $\Pr_{\vec{T}}[f_{\sigma_s}(\lat+\vec{T}) \leq \frac{1}{e^4V(\lat)}]
		\leq e^{-4}$.
	\end{enumerate}
	
	Since $1/6 + e^{-4/3} + 2e^{-4} \leq 1/2$, by the union bound we have $\vec{T}$
	satisfies the complement of all the above events with probability at least
	$1/2$. Let $\vec{t} \in \R^n$ be such a setting of $\vec{T}$. Clearly, the
	complement of (a) implies that the decoding error is at most $6\eps$, and the
	complement of (b) implies that averaged squared power per symbol is at most
	$\frac{1}{n}(n + 4\sqrt{n}) \sigma_s^2 = (1+4/\sqrt{n})\sigma_s^2$.
	
	It now remains to show the gap to capacity, i.e.~a lower bound on the entropy of
	$D_{\lat+\vec{t},\sigma_s}$. By \eqref{eq:entropy} and the complement of
	(c) and (d), we have that
	\begin{align*}
	\frac{1}{n} \mathbb{H}(\vec{X} | \vec{T} = \vec{t}) &=
	\frac{1}{n} \log\left((\sqrt{2\pi \sigma^2_s})^n f_{\sigma_s}(\lat+\vec{t})\right)
	+ \frac{1}{2n} \E_{\vec{Z} \sim
		D_{\lat+\vec{t},\sigma_s}}[\|\vec{Z}/\sigma_s\|^2] \\
	&\geq \log \left(\sqrt{2\pi \sigma^2_s}/(e^{4/n}V(\lat)^{1/n})\right) +
	\frac{1}{2}(1-4/\sqrt{n}) \\
	&= \frac{1}{2}\log \left(2\pi e \sigma_s^2/V(\lat)^{2/n}\right)
	-(2/\sqrt{n}+4/n) \\
	&= \frac{1}{2}\log \left(\sigma_s^2 \inver(\lat)^{2} / \gamma(\Lambda,\eps))\right)
	-(2/\sqrt{n}+4/n) \\
	&= \frac{1}{2}\log \left(1+\sigma_s^2/\sigma_w^2\right)
	-(1/2\log \gamma(\Lambda,\eps)+2/\sqrt{n}+4/n) \quad \left(~\sigma_{\rm eff} \cdot
	\inver_\eps(\lat) = 1~\right) \text{ .}
	\end{align*}
	The theorem now follows recalling that the capacity of the Gaussian
	channel is $\frac{1}{2}\log(1+\sigma_s^2/\sigma_w^2)$.
\end{proof}

\section{Discussion}
\label{sec:discussion}
\subsection{Finite Blocklength}
Next we discuss how the behavior of the proposed scheme compares to the best finite-blocklength codes for the Gaussian channel. The optimal rate to which a code can converge to capacity is given, in nats per dimension, by \cite{PPV09}:
\begin{equation} R = C - \sqrt{\frac{V}{n}} Q^{-1}(\varepsilon) + O\left(\frac{\log n}{n}\right)
\end{equation}
where
$$V = \lim_{\varepsilon \to 0} \limsup_{n\to \infty} \frac{-n(C - R)^2}{2\log
\varepsilon} = \frac{1}{2}\left(1-\frac{1}{(\snr+1)^2}\right) $$
is called the \textit{dispersion} of the channel and $\varepsilon$ is the target
probability of error. As usual $Q(x) = \int_x^\infty e^{-t^2/2}/\sqrt{2\pi} dt$,
 $x \in \R$,
is the Gaussian tail distribution. Notice that $V \approx 1/2$ for high
signal-to-noise ratio.  There is an analogous result for unconstrained
constellations, in which case the dispersion is exactly $1/2$. If
$\delta_\varepsilon(n)$ denotes the maximum NLD of a constellation for which the
probability of error is at most $\varepsilon$, then \cite{IZF13}:
\begin{equation} \delta_\varepsilon(n) = \delta^* - \sqrt{\frac{1}{2n}} Q^{-1}(\varepsilon) + O\left(\frac{\log n}{n}\right),
\end{equation}
where $\delta^* = -\frac{1}{2} \log\left(2 \pi e \sigma_\eff^2\right)$ \danote{I
think in our application, we want this lattice to be good for the effective
error $\sigma_\eff$ as opposed to $\sigma_w$.} is the
Poltyrev limit. The value $\delta_\varepsilon(n)$ also dictates the optimal
behavior of an AWGN-good sequence of lattices. Indeed,
\begin{equation}\frac{1}{2} \log \gamma(\lat,\varepsilon) \geq \delta^{*} -\delta_{\varepsilon}(n)  = \sqrt{\frac{1}{2n}} Q^{-1}(\varepsilon) + O\left(\frac{\log n}{n}\right).
\end{equation}
Therefore for high $\snr$ the lattice Gaussian scheme approaches optimal dispersion, provided that it is coupled with an infinite lattice with optimal behavior.
We further observe that \cite[p 131]{ZamirBook} conjectured that the gap to capacity of non-equiprobable signalling is upper bounded by
$$\frac{1}{2} \log \left( \mu(\lat,\varepsilon) G(\lat) \right),$$
where $G(\lat)$ is the \textit{normalized-second moment} of $\lat$. Noting that $G(\lat) > 1/2\pi e$, up to lower order terms which vanish as $n \to \infty$, the expression for the gap to capacity in Theorem \ref{thm:awgn-good-enough} supports the conjecture.

\subsection{Relation to the Smoothing Parameter}

We now record a new connection between AWGN-Goodness of a lattice and the
smoothing parameter of its \emph{dual lattice}. As explained in the last two
sections, a useful quantitative measure of the AWGN-Goodness of a lattice $\lat$
with respect to a desired target error $\eps$ is its \emph{normalized volume to
noise ratio} $\gamma(\lat,\eps)$, where we recall that up to lower order terms,
$\ln(\gamma(\lat,\eps))/2$ upper bounds the gap to capacity of the discrete
Gaussian coding scheme on $\lat$ (appropriately scaled and shifted).   We also recall the definition of \textit{dual} lattice
$$\lat^* = \set{ \vec{x} \in \R^n : \left\langle \vec{x}, \vec{y} \right\rangle \in \mathbb{Z} \mbox{ for all } \vec{y} \in \lat}.$$

The smoothing parameter $\eta_\eps(\lat^*)$ of $\lat^*$~\footnote{We normalize the smoothing
parameter with $1/2$ in the exponent here instead of the usual $\pi$ for
simplicity.} is defined to be the unique scaling $s > 0$ such that
\[
\sum_{x \in \lat \setminus \set{0}} e^{-\|sx\|^2/2} = \varepsilon \text{ .}
\]
At a high level, the discrete Gaussian over $\lat^*$ sampled \emph{above the
smoothing parameter} ``looks like'' a continuous Gaussian, which justifies the
name.

The connection to the inverse error function was given in~\cite{CDLP13}, in the
context of understanding the complexity of approximating the smoothing
parameter. It can be expressed as follows:

\begin{lem}\cite{CDLP13}
For $\varepsilon \in [0,1]$ and $\lat \subset \R^n$ an $n$-dimensional lattice, we have that:
\[
\eta_{\eps/(1-\eps)}(\lat^*) \leq {\rm err}^{-1}_\eps(\lat) \leq 2 \cdot \eta_{\eps}(\lat^*) \text{ .}
\]
\end{lem}

The right hand side corresponds essentially to the union bound, though over the
entire lattice instead of just the facets of the Voronoi cell, which should be
tight when $\eps$ is very small (say exponentially small in the lattice
dimension). The left hand side inequality is derived by estimating the Gaussian
mass of lattice shifts of the Voronoi cell, and for this side it is unclear to
us when it can be tight. Unfortunately, even for random lattices both sides can
fail to be tight. In particular, for an $n$-dimensional random lattice $\lat$ of
determinant $1$ and fixed constant error probability $\eps$, $\eta_\eps(\lat^*)
\approx \sqrt{2\pi}$, whereas ${\rm err}^{-1}_\eps(\lat) \approx \sqrt{2\pi e}$,
thus the inverse error function is a $\sqrt{e}$ factor bigger than the dual
smoothing parameter. Nevertheless, the above approximate characterization gives
a new and possibly easier way to check design criterion for good coding lattices.

We leave it as an open problem to understand whether there is a tighter
connection between the smoothing parameter and the inverse error function. In
particular, an interesting question is whether $\eta_\eps(\lat^*) \approx
\sqrt{2 \pi} V(\lat)^{-1/n}$ \footnote{It is not hard to check via the Poisson
summation formula that this is in fact a lower bound on the smoothing
parameter.} implies that ${\rm err}^{-1}_\eps(\lat) \approx \sqrt{2 \pi e}
V(\lat)^{-1/n}$.
\anote{I think it should be $V(\lat)^{-1/n}$ for both terms}
\danote{You're right. Fixed it.}
\danote{decided not to philosophize too much about the relationship with LB14
because the precise relationship between what we need for coding and the
smoothing parameter on either the primal or dual side is now quite unclear to me}

\subsection{Finite Dither}

Our main Theorem \ref{thm:awgn-good-enough} relies on a probabilistic argument
based on the generation of a continuous variable $\mathbf{T}$, that plays the
role of a dither, along with the generation of a discrete Gaussian $D_{\lat +
\mathbf{T}, \sigma_s}$. The usage of continuous dithering is undesirable in
practice since it may require sharing an infinite number of random bits at the
transmitter and receiver. In this section we discuss the possibility of choosing
the dither from a \textit{finite} set of possible vectors. The key idea is to
choose a sufficiently fine lattice $\lat' \supset \lat$ so that the distribution
$D_{\lat^\prime, \sigma_s}$ is sufficiently close to
$\mathcal{N}(0,\sigma_s^2)$. For this purpose, we need a finite version of the
sampling Lemma \ref{lem:sampling}, where we choose a dither from cosets
$\lat'/\lat$ with distribution induced by $D_{\lat', \sigma_s}$. Its proof is
similar to that of Lemma \ref{lem:sampling} and is included here for clarity.

\begin{lem}[Discrete Sampling Lemma]
Let $\lat \subset \lat'$, $\mathbf{T} \sim  D_{\lat', \sigma_s}$, and
$\mathbf{T}' = \mathbf{T} \text{ mod } \lat$. If $\mathbf{X}$ is Sampled from
distribution $D_{\lat + \mathbf{T}', \sigma_s}$, then $\mathbf{T}$ and $\mathbf{X}$ are identically distributed.
\label{lem:discrete-sampling-lemma}
\end{lem}

\begin{proof}
Analogously, we need show that $\vec{X}$ and $\vec{T}$ have the same probability distribution. For any set $A
\subseteq \lat'$, we have that
\begin{align*}
\Pr[\vec{X} \in A] &= \sum_{\vec{t}\in \lat'/\lat} \Pr[\vec{X} \in A | \mathbf{T}' = \vec{t}] \Pr[\mathbf{T}' = \vec{t}] \\
      &= \sum_{\vec{t}\in \lat'/\lat}
       \frac{f_{\sigma_s}(A \cap (\lat+\vec{t}))}{f_{\sigma_s}(\lat+\vec{t})}
\times \sum_{\vec{x} \in \lat+\vec{t}}
\frac{f_{\sigma_s}(\vec{x})}{f_{\sigma_s}(\lat')} \\
     &= \sum_{\vec{t}\in \lat'/\lat}
       \frac{f_{\sigma_s}(A \cap (\lat+\vec{t}))}{f_{\sigma_s}(\lat')}  \\
     &=  \frac{f_{\sigma_s}(A \cap \lat')}{f_{\sigma_s}(\lat')}
      = \Pr[\vec{T} \in A]  \text{,}
\end{align*}
as desired.
\end{proof}

To see how fine $\lat'$ should be, we need the definition of \textit{flatness factor}, which quantifies the distance between the uniform distribution in the Voronoi cell of a lattice and the sum of Gaussians over its cosets.

\begin{deft}[Flatness Factor, \cite{LB14}]
For a lattice $\lat$ and a parameter $\sigma$ we define the \textit{flatness factor} as
\begin{equation}
\epsilon_{\lat}(\sigma) \triangleq \max_{\xx \in \mathcal{V}(\lat)} \left|V(\lat) f_{\sigma}(\lat+\xx) - 1 \right|.
\end{equation}
\end{deft}


We note that there is an inverse relationship with the smoothing parameter, that
is for $\varepsilon,\sigma > 0$, we have $\epsilon_{\lat}(\sigma) = \varepsilon
\Leftrightarrow \eta_{\varepsilon}(\lat) = 2\pi \sigma$. The following lemma
shows that equivalent noise is approximately Gaussian if the flatness factor
of $\lat'$ is small, and in all cases, satisfies a strong tail bound.

\begin{lem} Let $\mathbf{X}$ be as in Lemma
\ref{lem:discrete-sampling-lemma} and $\vec{W}_{\eff} = (1-\alpha)\vec{X} +
\alpha \vec{W}$ be the effective noise, where $\vec{W} \sim
\mathcal{N}(0,\sigma_w^2)$ and let $g(\vec{w}_{\eff})$ be the pdf of
$\vec{W}_{\eff}$. Then the following holds:

\begin{enumerate}
\item $\forall \vec{w}_\eff \in \R^n$,
$g(\vec{w}_\eff) \leq (1+\epsilon_{\lat'}(\sqrt{\alpha}\sigma_s)) \cdot
f_{\sigma_{\eff}}(\vec{w}_{\eff})$.
\item $\forall \eps \in (0,1)$, $\Pr[\|\vec{W}_\eff\| > \sqrt{(1+\eps)n}
\sigma_\eff] \leq e^{-\frac{n}{4}(\eps^2-\eps^3)}$.
\label{lem:discrete-dither-error-bnd}
\end{enumerate}
\end{lem}
\begin{proof}

We prove part $1$. By Lemma \ref{lem:convolution}, proven in the
appendix, we have that $g(\vec{w}_\eff) = f_{\sigma_{\eff}}(\vec{w}_\eff)
\frac{f_{\sqrt{\alpha}\sigma_s}(\lat' + \vec{w}_\eff)}{f_{\sigma_s}(\lat')}$. By
the definition of the flatness factor, we have that
$f_{\sqrt{\alpha}\sigma_s}(\lat') \leq
(1+\epsilon_{\lat'}(\sqrt{\alpha}\sigma_\eff))/\det(\lat')$. Furthermore, by the
Poisson summation formula, it easily follows that $f_{\sigma_s}(\lat') \geq
1/\det(\lat')$. The combination of these two estimates yields the bound.   

The tail bound in part $2$ is proved in the full version of \cite{CLB16}. We
sketch the proof here for completeness. The tail bound relies on properties of
sub-Gaussian random variables. We recall that a random vector $\vec Z \in \R^n$ is
sub-Gaussian with parameter $\sigma$ if $\E[e^{\pr{\vec y}{\vec Z}}] \leq
e^{\sigma^2/2 \|\vec y\|^2}$ for all $\vec y \in \R^n$. In particular, since
$\vec{W}$ is $\mathcal{N}(\vec 0,\sigma_w^2)$, it is sub-Gaussian with parameter
$\sigma_w$. Banaszczyk~\cite[Lemma 2.4]{Banaszczyk95} established that centered
discrete Gaussian random vectors are sub-Gaussian, in particular, that $\vec X
\sim D_{\Lambda',\sigma_s}$ is sub-Gaussian with parameter $\sigma_s$. By the 
additivity properties of sub-Gaussian random vectors, we conclude that
$\vec{W}_\eff = (1-\alpha) \vec X + \alpha \vec W$ is sub-Gaussian with
parameter $\sqrt{(1-\alpha)^2 \sigma_s^2 + \alpha^2 \sigma_w^2} = \sigma_\eff$.
For a $\sigma_\eff$ sub-Gaussian vector $\vec W_{\eff}$, a tail bound due Hsu,
Kakade and Zhang~\cite[Theorem 2.1]{HKZ12} establishes that for all $t > 0$, we
have
\[
\Pr[\|\vec W_{\eff}\|^2 > \sigma_{\eff}^2(n + 2 \sqrt{n t} + 2 t)] \leq e^{-t} .
\]  
The desired bound now follows by plugging in $t = \frac{n}{4} (\sqrt{1+2
\eps}-1)^2 > \frac{n}{4}(\eps^2-\eps^3)$.
\end{proof}

To recover Theorem~\ref{thm:awgn-good-enough} with the continuous dither
replaced by the discrete one, it in fact suffices to show that the probability
bounds given in
Lemmas~\ref{lem:decoding-error},~\ref{lem:concentration},~\ref{lem:gaussian-lb}
hold for the discrete dither distribution $\vec{T} \in \lat'$ up to additional
$1+o(1)$ factors. We presently explain how the conclusions of each Lemma differs
for the discrete dither. By inspecting the proof of
Lemma~\ref{lem:decoding-error}, using Lemma~\ref{lem:discrete-dither-error-bnd}
one can conclude that the error probability bound increases to
$(1+\epsilon_{\lat'}(\sqrt{\alpha}\sigma_s))/\gamma$. The conclusions in
Lemma~\ref{lem:gaussian-lb} and Lemma~\ref{lem:concentration} part $1$ (upper
tail bound) remain unchanged, using the inequality $f_{\sigma_s}(\lat') \geq
1/\det(\lat')$ (by Poisson summation) to prove Lemma~\ref{lem:gaussian-lb} and
Banaszczyk's discrete Gaussian upper tail bound \cite[Lemma 1.5]{Banaszczyk93}
to prove Lemma~\ref{lem:concentration} part $1$. Lemma~\ref{lem:concentration}
part $2$ (lower tail bound) degrades by a factor
$1+\epsilon_{\lat'}(\sigma_s/\sqrt{1+\eps/(2-\eps)})$, where we simply use the
flatness factor to bound the relevant moment generating function. Given the
above, inspecting the proof of Theorem $1$, the conclusions remain
asymptotically identical, as $n \rightarrow \infty$, when
\begin{equation}
\label{eq:dither-needs}
\max
\set{\epsilon_{\lat'}(\sqrt{\alpha} \sigma_s),
\epsilon_{\lat'}(\frac{\sigma_s}{1+2/\sqrt{n}})} \rightarrow 0 .
\end{equation}
Since $\alpha := \frac{\sigma_s^2}{\sigma_s^2 + \sigma_w^2} < 1$ the first term
in the maximum is asymptotically dominant. We recall that the first term helps
control the decoding error probability, using
Lemma~\ref{lem:discrete-dither-error-bnd}, and the second helps lower bound the
entropy of the coding distribution, using the lower tail bound in
Lemma~\ref{lem:concentration}. 

While the above bounds may seem sharp, we suspect that the flatness bound on the
dither distribution we require for achieving the desired decoding error
probability may be unnecessary. This is indeed true if we relax the requirement
that the decoding error be bounded by the probability that a variance
$\sigma_{\eff}^2$ Gaussian random vector falls outside the Voronoi cell of
$\lat$ (i.e.~a very fine grained requirement), to instead ask for a vanishing
error probability when the lattice $\lat$ is chosen appropriately. In
particular, it was shown in~\cite[Section C]{OE16} that a random mod-$p$ lattice
$\lat$, with any fixed normalized volume $\det(\lat)^{2/n} > 2\pi e
\sigma_\eff^2$, has vanishing decoding error whenever the equivalent error falls
outside of the $\sqrt{n}\sigma_\eff$ noise sphere with negligible probability.
Lemma~\ref{lem:discrete-dither-error-bnd} shows that the required tail bound
indeed holds for the equivalent error $\vec W_\eff$ in our setting, where $\Pr[\|\vec
W_\eff\| > \sqrt{(1+n^{-1/4})n}\sigma_\eff] \leq
e^{-\sqrt{n}/8}$ for $n \geq 16$. Note that for
such lattices, i.e.~where the error probability is essentially independent of
the dithering distribution, we can relax requirement~\ref{eq:dither-needs} to
just $\epsilon_{\lat'}(\sigma_s/(1+2/\sqrt{n})) \rightarrow 0$.  

%

We note that the fine lattice $\lat'$ will satisfy the
requirement~\ref{eq:dither-needs}, or the above relaxation, as long as it is
sufficiently dense. In particular, $\lat'$ can simply be chosen to be a
scaled-down version of $\lat$. However, the number of random bits necessary to
generate the dither distribution in $\lat'/\lat$ does depend on the properties
of $\lat'$. In particular, it is minimized if $\lat'$ is chosen from a
\textit{flatness-good} family of lattices. \cite[Theorem 1]{LLBS_12} shows that
a random mod-$p$ lattice $\lat'$ will have vanishing flatness factor
$\epsilon_{\lat'}(\sigma_{s}/(1+2/\sqrt{n}))$ for any fixed normalized volume
$V(\lat')^{2/n} < 2 \pi \sigma_s^2$ (we restrict to the relaxed requirement).
Furthermore, as explained above, a random mod-$p$ lattice $\lat$ has vanishing
error probability against sub-Gaussian effective error with parameter
$\sigma_\eff$, for any fixed normalized volume $V(\lat)^{2/n} > 2 \pi e
\sigma_{\eff}^2$. Note that if $2 \pi \sigma_s^2 > 2 \pi e \sigma_{\eff}^2
\Leftrightarrow \snr > e-1$, we can simply let $\lat' = \lat$, i.e.~no dither is
necessary (recovering the result of~\cite{LB14}). If $\snr \leq e-1$, we get a
dithering rate bounded by
\[
R' \leq \frac{1}{n} \log\frac{V(\lat)}{V(\lat')} \leq \frac{1}{2}(1-\log(1+\snr)) + \delta_n \quad \mathrm{nats/channel\text{ }use}
\]
for some $\delta_n \to 0$ as $n\to \infty$. Note that to make the bounds
rigorous, we must in fact enforce that $\lat \subseteq \lat'$. Though this is
not directly addressed above, it is easy to achieve by standard
nesting techniques (see for example~\cite{OE16}). Alternatively, as mentioned
above, one could simply use two different scalings of the same random mod-$p$
lattice to achieve the desired nesting. This will however enforce a rate of
$\log 2$, since we must scale the lattice down by an integer and $\lceil
\sqrt{e/(1+\snr)} \rceil = 2$ .   

Since it can be challenging to construct {flatness-good} families of lattices,
in practice it may be convenient to use a scaled-down version of a simple
lattice. This technique was previously used in polar lattices \cite{YanLLW14}.

{\bf Case study: polar lattices.} With discrete Gaussian shaping, polar lattices
achieve the capacity of the Gaussian channel for any SNR \cite{YanLLW14}. Here,
$\lat$ is an AWGN-good lattice from Construction D, while $\lat'=a\mathbb{Z}^n$
for some small enough scaling factor $a$. \cite{YanLLW14} uses $a=c\cdot
\sigma_{\mathrm{eff}}/\sqrt{\log n}$, for an appropriately small constant $c$, to
obtain a vanishing flatness factor $\epsilon_{\lat'}(\sigma_{\mathrm{eff}})$. In
polar lattices, a polar code is employed on each level, where random bits in the
frozen set play the role of a dither.

Since $V(\lat')=a^n$, the rate of dithering bits is bounded by
\[
R' \leq \frac{1}{n} \log\frac{V(\lat)}{V(\lat')} \leq \frac{1}{2}\log\left(\frac{2\pi e}{(1/\sqrt{\log n})^2}\right) + \delta_n = O(\log \log n).
\]
Note that basically this is due to the fact that $O(\log \log n)$ levels are
used. However, in practice, a small number of levels are enough, especially at
low SNR. Therefore, the rate of dithering bits is essentially a small constant.
We note that the dithering rate for polar lattices is higher, due to a
suboptimal $\Lambda'$.

\subsection{Peak Power}

Let $X_i$ ($i=1,2,\ldots,n$) denote the $i$-th component of the codeword
$\mathbf{X}$. With continuous Gaussian dithering, the average power
$\E[X_i^2]=\sigma_s^2$ since $X_i$ is exactly Gaussian. With discrete Gaussian
dithering, we still have $\E[X_i^2]\approx \sigma_s^2$ if
$\epsilon_{\lat'}(\sigma_s)$ is small \cite[Lemma 5]{LB14}. However, DPS may
raise the issue of high peak power since $X_i$ is random and has infinite support. The same issue arose
in \cite{QFH17} where discrete dithering was used. In \cite{QFH17} (see also
\cite{LB14}), this issue was solved by simply sending a zero vector if the
signal power is too large, i.e., if $\|\mathbf{X}\|^2 > nP$. This essentially
converts an encoding failure into a decoding error. Its drawback is that the
code will lose linearity.

For lattice Gaussian distribution $D_{\lat',\sigma_s}$, one may give a bound on
$\Pr(\|\vec X\|^2 > tn\sigma_s^2)$ for any $t\geq 1$ (see
\eqref{eq:boundGaussian}). Yet in engineering it is more reasonable to consider
the peak power of each component.  Thus we apply a modular operation to each
$X_i$ in order to limit the peak power, namely, we send $\bar{X_i} = X_i \mod B$
for certain modulus $B$, which was previously used in polar lattices
\cite{YanLLW14}. For this purpose, we need a bound on $\Pr(|X_i| >
t\sigma_s)$ for some margin $t>1$. Such margin bounds may be derived directly
from the fact that $\vec X \sim D_{\lat',\sigma_s}$ is sub-Gaussian with
parameter $\sigma_s$ (see~\cite[Lemma 2.4]{Banaszczyk95}), from which we derive
the tail bound
\begin{equation}
\label{eq:boundGaussian1D}
\Pr[|X_i| > t \sigma_s] \leq 2e^{-t^2/2}.
\end{equation}



By the union bound, we have that the probability of $\vec X$ exceeding peak
power $t^2 \sigma_s^2$ is bounded by
\begin{equation}
\Pr\left( \bigcup_{i=1}^n \{X_i^2 \geq t^2\sigma_s^2\} \right) \leq 2ne^{-t^2/2}.
\end{equation}
Thus, it is possible to choose $t = O(\sqrt{\log n})$ such that the above
probability tends to zero. Recall that in reality, we need such a bound to hold
for the coding distribution {\em after} conditioning on the value of the dither.
By Lemma~\ref{lem:discrete-sampling-lemma} (discrete sampling lemma) however, the
above probability is in fact an upper bound on the average probability that
the coding distribution exceeds the desired peak power, where the average is taken
over the choice of dither. Therefore, by Markov's inequality, with probability
at least $1-\eps$ over the choice of dither, the probability that the coding
distribution exceeds peak power $t^2\sigma_s^2$ is at most $2ne^{-t^2/2}/\eps
\rightarrow 0$, for $t = O(\sqrt{\log n})$ and any fixed $\eps \in (0,1)$.   

From the perspective of the afore-mentioned modulus, we can now
choose $B=O(\sqrt{\log n})\sigma_s$. For the modulo operation to maintain
linearity however, we must of course require that $B \mathbb{Z}^n \subset \lat$,
where $\lat$ is the coding lattice. Typical ``Construction A'' lattices (e.g. the ones in \cite{EZ04}) satisfy this requirement, as well as polar lattices (however polar lattices need a larger value of $B$, namely $B = (\log n)^{O(1)}
\sigma_s$, due to the $O(\log \log n)$ number of levels used in their construction).

\section{Converse}
\label{sec:converse}

In Section \ref{sec:AWGN-good}, we have shown that the shifted (or ``dithered'') lattice Gaussian is capacity-achieving. To close this paper, we will argue that for very low signal-to-noise ratio the shift is essentially necessary. We will prove the following ``converse'':

\begin{thm} Let $(\Lambda_n)_{n=0}^{\infty}$ be a sequence of lattice Gaussian codes, with corresponding distributions $D_{\Lambda_n,\sigma_s=\sqrt{P}}$. If $\sigma_s^2/\sigma_w^2 < 1$ and
	\begin{equation}
	\lim_{n\to \infty} \frac{1}{n}\mathbb{H}(\XX_n) = C(\text{\upshape{SNR}}),
	\end{equation}
	then the probability of error $P_e(\Lambda_n)$ of the maximum-a-posteriori decoder is bounded away from $0$.
	\label{thm:converse}
\end{thm}

Before exhibiting the proof, we provide a heuristic argument that justifies why Theorem \ref{thm:converse} should be true for random lattices. The average of the Gaussian sum $f_{\sigma_s}(\Lambda)$ over a typical ``random'' lattice of volume $V$ satisfies (e.g. \cite[Lemma 3]{LLBS_12}):
\begin{equation}\label{eq:Gaussian-mean}
\E \left[f_{\sigma_s}(\Lambda)\right] = (2\pi\sigma_s)^{-n/2} + V^{-1} \int f_{\sigma}(\xx) {\rm dx}
= (2\pi\sigma_s)^{-n/2} + V^{-1}.
\end{equation}

 For a sequence of lattices to have vanishing probability of error, the Poltyrev limit for the best NLD \eqref{eq:poltyrevLimit} implies that $V(\Lambda)^{2/n} > (2\pi e) \sigma_{\text{eff}}^2$ where $\sigma_{\text{eff}}^2 = \sigma_s^2 \sigma_w^2/(\sigma_s^2+\sigma_w^2)$ is the effective noise power. Under this condition, the typical distribution $D_{\Lambda,\sigma_s}$ of a random lattice has essentially all mass is concentrated in the zero vector. Indeed, we have:
\begin{equation*}
\begin{split}
\E[P(\XX=0)]&=\E[(\sqrt{2\pi\sigma_s^2})^{-n} /f_{\sigma_s}(\Lambda)]\geq (\sqrt{2\pi\sigma_s^2})^{-n}/\E[f_{\sigma_s}(\Lambda)] \\
&\stackrel{(a)}{=}  (\sqrt{2\pi\sigma_s^2})^{-n} \left(\frac{(\sqrt{2\pi \sigma_s^2})^n}{V(\Lambda)^{-1} (\sqrt{2\pi \sigma_s^2})^n+1}\right) \geq \frac{1}{\left(\frac{1}{e}\left(1+\frac{\sigma_s^2}{\sigma_w^2}\right)\right)^{\frac{n}{2}}+1} \to 1,
\end{split}
\end{equation*}
where (a) is due to \eqref{eq:Gaussian-mean} and the last limit holds for $\snr < e-1$. In this case, it is not hard to see that the entropies of $\mathcal{D}_{\Lambda,\sigma_s}$ necessarily tends to zero, therefore no positive rate is achievable. Following this intuition, we prove Theorem \ref{thm:converse} by showing that a centered lattice Gaussian with vanishing error probability has essentially all mass concentrated in the origin.

We recall the definition of the effective noise $\vec{W}_{\eff} = (\alpha-1) \vec{X} + \alpha \vec{W}$, where $\alpha = \sigma_s^2/(\sigma_s^2+\sigma_w^2)$ is the MMSE coefficient. Following \cite[Claim 3.9]{Regev05}, we obtain the distribution of $\ww_{\eff}$ in the following lemma. We give a proof in the Appendix for completeness.
\begin{lem}
	The probability density function $g$ of the effective noise $\vec{W}_{\eff}$ is given by:
	$$g(\ww_{\eff}) = \frac{e^{\frac{-\left\|\ww_{\eff}\right\|^{2}}{2 \sigma_{\eff}^2}}}{\left(\sqrt{2\pi}\sigma_{\eff}\right)^{n}} \frac{f_{\sqrt{\alpha} \sigma_s}(\Lambda+\ww_{\eff})}{f_{\sigma_s}(\Lambda)},$$
	where $\sigma_{\eff}^2 = \sigma_w^2 \sigma_s^2/(\sigma_w^2+\sigma_s^2)$ is the variance of $\ww_{\eff}$.
	\label{lem:convolution}
\end{lem}
\begin{lem}[Relation between entropy and the mass of zero]
	\label{lem:entropy} Let $P_0(\Lambda)$ denote the probability that $\XX = 0$, where $\XX \sim \mathcal{D}_{\Lambda,\sigma_s}$. We have
	$$\frac{1}{n}\mathbb{H}(\XX) \leq -\frac{1}{n}\log(P_0(\Lambda)) + \pi (1-P_0(\Lambda)) + \frac{1.8 e^{-1.7n}}{n} .$$
	\label{lemma:entropy}
\end{lem}
\begin{proof}
	We will show the more general bound
	\begin{equation}
	\frac{1}{n}\mathbb{H}(\XX) \leq -\frac{1}{n}\log P_0(\Lambda) + {\alpha} (1-P_0(\Lambda)) + \frac{1}{n\phi(\alpha)} e^{-\alpha \phi(\alpha) n},
	\label{eq:boundEntropy}
	\end{equation}
	for $\phi(\alpha) = 1 - (1/2\alpha)\log( 2\alpha e)$ and any $\alpha > 1$.  The lemma will follow by taking $\alpha = \pi$.
	
	From the definition of entropy,
	\begin{equation}\frac{1}{n}\mathbb{H}(\XX) = \frac{1}{n}\log\left ( (\sqrt{2\pi} \sigma_s)^n f_{\sigma_s}(\Lambda)\right)+ \frac{1}{2n\sigma_s^2}E\left[\left\| {\XX} \right\|^2 \right].
	\label{eq:defEntropy}
	\end{equation}
	First notice that $P_0(\Lambda) = ((\sqrt{2\pi} \sigma_s)^n f_{\sigma_s}(\Lambda))^{-1 }$, therefore the terms inside the logarithms in Eqs. \eqref{eq:boundEntropy} and \eqref{eq:defEntropy} coincide. Now using \cite[Lemma 2.13]{DRS14} (cf. \cite[Lemma 8]{LB14}), for $t \geq 1$, we have
\begin{equation}\label{eq:boundGaussian}
\Pr\left(\frac{1}{2n\sigma_s^2}\left\| \XX\right\|^2 \geq t \right) \leq e^{-nt + \frac{n}{2} \log (2t e )}.
\end{equation}
	Moreover $\Pr\left(\left\| \XX\right\|^2 \geq 2nt\sigma_s^2 \right) \leq 1-P_0(\Lambda)$. Therefore we obtain bound
	\begin{equation}\begin{split}E\left[\left\| {\XX}/\sqrt{2n\sigma_s^2} \right\|^2 \right] &= \int_{0}^{\infty } \Pr\left(\frac{1}{2n\sigma_s^2}\left\| \XX \right\|^2 \geq t \right) dt\\
&\leq (1-P_0(\Lambda))\alpha + \int_{\alpha}^{\infty } \Pr\left(\frac{1}{2n\sigma_s^2}\left\| \XX \right\|^2 \geq t \right) dt \\
	&\leq (1-P_0(\Lambda))\alpha + \int_{\alpha }^{\infty } e^{-nt \phi(\alpha)} dt
	\end{split}
	\end{equation}
	Evaluating the integral gives us \eqref{eq:boundEntropy}.
	
\end{proof}

\indent \textbf{Probability of Error Analysis.}
Lemma \ref{lem:convolution} allows us to relate the probability that the effective noise lies outside the Voronoi cell of a lattice and the Gaussian mass of the point $\mathbf{0} \in \Lambda$. For this purpose we first note that Lemma \ref{lem:convolution} implies the following relations:
\begin{equation}
f_{\sigma_s}(\Lambda) = \int_{\mathbb{R}^n} \frac{e^{-\left\| \ww \right\|^2/2{\sigma_{\eff}}^2}}{(\sqrt{2 \pi} \sigma_{\eff})^n} f_{\sqrt{\alpha}\sigma_s}(\Lambda + \ww) d\ww = \int_{\mathcal{V}(\Lambda)}  f_{{\sigma_{\eff}}}(\Lambda + \ww) f_{\sqrt{\alpha}\sigma_s
}(\Lambda + \ww) d\ww.
\label{eq:DDMagic}
\end{equation}
and
$$P_e(\Lambda) f_{\sigma_s}(\Lambda) = \int_{\mathcal{V}(\Lambda)}f_{\sigma_\eff}(\Lambda \nozero + \ww)  f_{\sqrt{\alpha}\sigma_s}(\Lambda + \ww) d\ww.$$

\noindent Recalling that $P_0(\Lambda) =((\sqrt{2 \pi}\sigma_s)^n f_{\sigma_s}(\Lambda))^{-1}$, the above equality indeed relates the probability of the discrete Gaussian hitting zero and the probability of error. We will first bound $f_{\sigma_s}(\Lambda)$ which will imply a bound on $P_0(\Lambda)$ and on the entropy, by Lemma \ref{lemma:entropy}, showing the assertion in Theorem \ref{thm:converse}.

\textit{Proof of Theorem \ref{thm:converse}:}  We have
\begin{equation}\begin{split}f_{\sigma_s}(\Lambda_n) &= \int_{\mathcal{V}(\Lambda_n)}  f_{{\sigma_{\eff}}}(\Lambda_n\nozero + \ww) f_{\sqrt{\alpha}\sigma_s
}(\Lambda_n + \ww) d\ww +  \int_{\mathcal{V}(\Lambda_n)}  f_{{\sigma_{\eff}}}(\ww) f_{\sqrt{\alpha}\sigma_s
}(\ww) d\ww \\
&+ \int_{\mathcal{V}(\Lambda_n)} f_{{\sigma_{\eff}}}(\ww)f_{\sqrt{\alpha}\sigma_s
}(\Lambda_n\nozero + \ww) d\ww .
\end{split}
\label{eq:f1}
\end{equation}
We proceed to bound the three terms on the right-hand side of \eqref{eq:f1}. The first term is equal to $P_e(\Lambda_n)f_{\sigma_s}(\Lambda_n)$ while the second term satisfies
$$\int_{\mathcal{V}(\Lambda_n)}  f_{{\sigma_{\eff}}}(\ww) f_{\sqrt{\alpha}\sigma_s
}(\ww) d\ww \leq \int_{\mathbb{R}^n}  f_{{\sigma_{\eff}}}(\ww) f_{\sqrt{\alpha}\sigma_s
}(\ww) d\ww = \frac{1}{(\sqrt{2\pi}\sigma_s)^n}.$$
Noting that $\sqrt{\alpha}\sigma_s= \sigma_{\eff}(\sigma_s/\sigma_w)$ and using the assumption $\sigma_s/\sigma_w < 1$, the last term can be upper bounded as
\begin{equation}
\begin{split}
\int_{\mathcal{V}(\Lambda_n)}  f_{{\sigma_{\eff}}}(\ww) f_{\sqrt{\alpha}\sigma_s
}(\Lambda_n\nozero + \ww) d\ww \leq \int_{\mathcal{V}(\Lambda_n)}  f_{{\sigma_{\eff}}}(\Lambda_n\nozero + \ww) f_{\sqrt{\alpha}\sigma_s
}(\ww) d\ww \leq P_e(\Lambda_n) f_{\sigma_s}(\Lambda_n).
\end{split}
\end{equation}
Combining altogether, we obtain the bound
$$f_{\sigma_s}(\Lambda_n) \leq 2 P_e(\Lambda_n) f_{\sigma_s}(\Lambda_n) + \frac{1}{(\sqrt{2\pi}\sigma_s)^n} \Rightarrow  P_e(\Lambda_n) \geq \frac{1}{2} (1 - P_0(\Lambda_n)).$$
This implies in turn that, if the probability of error $P_e(\Lambda_n) \to 0$, then $P_{0}(\Lambda_n) \to 1$ and, from Lemma \ref{lemma:entropy}, $(1/n)\mathbb{H}(\XX_n) \to 0$. Conversely, if we force $(1/n)\mathbb{H}(\xx)$ to tend to a positive value, $P_{0}(\Lambda_n)$ is bounded away from one, and therefore $P_e(\Lambda_n)$ is bounded away from zero. \qed

\section{Conclusion}
\label{sec:conclusion}
In this paper we have shown that DPS can convert any lattice which is good for the unconstrained AWGN channel into a good code in the power-constrained setting.
For instance, any sphere-bound achieving lattice in the sense of \cite{Forney00} can be coupled with our results in order to produce capacity-achieving codes. We stress the fact that previous schemes in the literature strictly need extra conditions other than AWGN-goodness, such as flatness-goodness or quantization-goodness.

We have further demonstrated the efficacy of discrete dithering in place of continuous Gaussian dithering. Optimizing the rate of random bits in discrete dithering for a broad range of dimensions and rates are an interesting further research direction. Improving the second-order analysis in order to achieve the right dispersion for all SNRs is a further point of interest and left as an open problem.

Finally, although the heuristic argument exhibited in Section \ref{sec:converse} reveals that the centered distribution should fail to achieve capacity for $\snr < e-1$, the actual proof only holds for $\snr < 1$, leaving inconclusive the values $\snr \in  [1,e-1]$. Furthermore, the scheme \cite{LB14} fixes the variance parameter $\sigma_s^2 = P$ a priori, whereas one could potentially achieve better rates by choosing $\sigma_s$ adaptively depending on the dimension. A stronger converse that can handle varying $\sigma_s$ would strengthen our results. We believe that this would require completely new arguments.

\section*{Acknowledgements}
This work was conceived during a visit of the first author (AC) to the second
one (DD), enabled by the NWO Veni grant 639.071.510. DD would like to thank
Divesh Aggarwal for useful discussions. The authors would like to thank Ram Zamir for clarifying the
role of dither in Voronoi constellations, Yair Yona for kindly providing the manuscript \cite{YH17} and the anonymous reviewers for their comments and suggestions.

\appendix
\label{sec:appendix}
In this appendix we prove Lemma \ref{lem:convolution}. Recall the MMSE coefficient $\alpha$ and the effective noise parameter $\sigma_\eff^2$:
$$\alpha = \frac{\sigma_s^2}{\sigma_s^2+\sigma_w^2}, \,\,\,\,\, \sigma_{\eff}^2 = \sigma_w^2 \frac{\sigma_s^2}{\sigma_w^2+\sigma_s^2} = \alpha \sigma_w^2.$$

Let $\tilde{\XX} = (\alpha-1) \XX$, where $\XX \sim D_{\lat,\sigma_s}$, then
we have that $\tilde{\XX} \sim D_{(1-\alpha)\lat, (1-\alpha)\sigma_s}$.
Indeed, the probability of picking $\tilde{\xx} \in (\alpha-1) \lat$, noting that $1-\alpha \in (0,1)$, is given by
$$ \frac{f_{\sigma_s(1-\alpha)}(\tilde{\xx})}{f_{\sigma_s(1-\alpha)}( (1-\alpha)
\lat)} = \frac{f_{\sigma_s}(\tilde{\xx}/(\alpha-1))}{f_{\sigma_s}(\lat)} . $$
Let $\tilde{\WW} = \alpha \WW \sim \mathcal{N}(0,\alpha^2 \sigma_s^2)$, and $\WW_{\eff} = \tilde{\XX} + \tilde{\WW}$. The distribution of the continuous variable $\WW_{\eff}$ is given by the convolution of the distributions of $\tilde{\XX}$ and $\tilde{\WW}$, namely
\begin{equation}\begin{split} g(\ww_{\eff}) &= \sum_{\tilde{\xx} \in
(\alpha-1)\lat} \frac{f_{\sigma_s}(\tilde{\xx}/(\alpha-1))}{f_{\sigma_s}( \lat)} f_{\sigma_w \alpha}(\ww_{\eff} - \tilde{\xx}) \\
&= \frac{1}{(\sqrt{2\pi \sigma_w^2 \alpha^2})^n} \times \frac{1}{(\sqrt{2\pi\sigma_s^2})^n f_{\sigma_s}(\lat)} \times \sum_{\tilde{\xx} \in \lat} e^{- \norm{\tilde{\xx}}^2/2\sigma_s^2 - \norm{\ww_{\eff} - (\alpha-1) \tilde{\xx}}^2/2\alpha^2\sigma_w^2}.
\\
& = \frac{1}{(\sqrt{2\pi \sigma_\eff^2})^n} \times \frac{1}{(\sqrt{2\pi\alpha \sigma_s^2})^n f_{\sigma_s}(\lat)} \times \sum_{\tilde{\xx} \in \lat} e^{- \norm{\tilde{\xx}}^2/2\sigma_s^2 - \norm{\ww_{\eff} - (\alpha-1) \tilde{\xx}}^2/2\alpha^2\sigma_w^2}
\end{split}
\end{equation}
To evaluate the exponents in the last expression we use the identity

\begin{equation*}
\frac{\norm{\tilde{\xx}}^2}{2\sigma_s^2} + \frac{\norm{\ww_{\eff} - (\alpha-1) \tilde{\xx}}^2}{2\alpha^2\sigma_w^2} = \frac{\norm{\ww_\eff}^2}{2 \sigma_\eff^2} + \frac{\norm{\tilde{\xx} +\ww_\eff}^2}{2\alpha \sigma_s^2}
\end{equation*}
which gives 
\begin{equation}
g(\ww_{\eff}) = \frac{e^{-\frac{\norm{\ww_\eff}^2}{2 \sigma_\eff^2}}}{(\sqrt{2\pi \sigma_\eff^2})^n} \times \frac{1}{(\sqrt{2\pi\alpha \sigma_s^2})^n f_{\sigma_s}(\lat)} \times \sum_{\tilde{\xx} \in \lat+ \ww_\eff} e^ {-\frac{\norm{\tilde{\xx}}^2}{2\alpha \sigma_s^2}}.
\end{equation}
This last expression for $g(\ww_\eff)$ coincides with the one in Lemma \ref{lem:convolution}.

\bibliographystyle{plain}
\bibliography{fix}
\end{document}